\documentclass[12pt]{article}
\usepackage{amssymb, amsmath, amsthm}
\usepackage{thmtools}
\usepackage{mathtools}
\declaretheorem[style=definition]{example}
\usepackage{appendix}
\usepackage{bm}
\usepackage{bbm}
\usepackage{caption}
\usepackage{enumerate} 
\usepackage{graphicx}
\usepackage{natbib}
\usepackage{setspace}
\usepackage{hyperref} 
\usepackage{color}
\usepackage{tikz}
\usepackage{mathrsfs}
\usetikzlibrary{snakes}
\usepackage{subfigure}
\numberwithin{equation}{section}
\usepackage[top=1.25in, bottom=1.25in, left=1.25in, right=1.25in]{geometry}

\DeclareMathOperator*{\supp}{supp}

\newtheorem{lem}{Lemma}

\newtheorem{thm}{Theorem}
\newtheorem*{thm*}{Theorem}
\newtheorem{cor}{Corollary}

\newtheorem{defn}{Definition}

\begin{document}
\renewcommand{\thefootnote}{\fnsymbol{footnote}}
\renewcommand\thmcontinues[1]{Continued}
\title{Intermediated Implementation}
\author{Anqi Li\footnote{Department of Economics, Washington University in St. Louis, anqili@wustl.edu. }
\and Yiqing Xing\footnote{Carey Business School, Johns Hopkins University,  xingyq@jhu.edu. We thank the associate editor, three anonymous referees, Chen Cheng, Matt Jackson, Xing Li, George Mailath, Mike Peters, Andy Postlewaite, Florian Scheuer, Ilya Segal, Chris Shannon, Anthony Lee Zhang, as well as the seminar audience at UC Berkeley, U of Michigan, UPenn, Stanford, ASSA, ES Summer Meetings and SED for helpful comments and suggestions. All errors are our own.}}
\date{\emph{Forthcoming in European Economic Review}}
\maketitle

\begin{abstract}
\noindent We examine problems of ``intermediated implementation,'' in which a single principal can only regulate limited aspects of the consumption bundles traded between intermediaries and agents with hidden characteristics. An example is sales, in which retailers offer  menus of consumption bundles to customers with hidden tastes, whereas a manufacturer with a potentially different goal from retailers' is limited to regulating sold consumption goods but not retail prices by legal barriers. We study how the principal can implement through intermediaries any social choice rule that is incentive compatible and individually rational for agents. We demonstrate the effectiveness of per-unit fee schedules and distribution regulations, which hinges on whether intermediaries have private or interdependent values. We give further applications to healthcare regulation and income redistribution.

\bigskip

\bigskip

\noindent Keywords: implementation; vertical structure; adverse selection; monopolistic screening.\\

\noindent JEL codes: D40, D82, H21, I13, L10. 

\bigskip 

\bigskip

\bigskip

\bigskip

\bigskip

\bigskip

\bigskip

\bigskip

\end{abstract} 
\renewcommand{\thefootnote}{\arabic{footnote}}
\pagebreak

\section{Introduction}\label{sec_intro}
We examine problems of ``intermediated implementation,'' which involve a principal, one or multiple intermediaries and a continuum of agents with hidden characteristics. Intermediaries offer menus of multifaceted consumption bundles to agents, whereas the principal with a potentially different goal from intermediaries' is limited to regulating sub-aspects of the sold consumption bundles by practical barriers. 

Problems that fit the above description abound. An example is car sales, in which dealers offer car-price bundles to customers with hidden tastes. A manufacturer, say, Mercedes Benz, has a potentially different goal from maximizing dealer profit (due to, e.g., additional concerns such as brand loyalty or the wedge introduced by warranties) but has only partial control over the sold car-price bundles. On the one hand, the manufacturer controls the car supply and actively tracks sold cars through their unique vehicle identification numbers (VIN). On the other hand, she cannot regulate prices explicitly or implicitly, as resale price maintenance is ruled illegal by the Sherman Act in the U.S. and the EU competition law.

In this paper, we take the barriers faced by the principal as given and examine how she can implement through intermediaries \emph{any} target social choice rule that is incentive compatible and individually rational for agents. We focus on quasi-linear environments in which intermediaries compete \`{a} la Bertrand and relax these assumptions in later sections. Our solutions exploit the essential properties of incentive compatible social choice rules rather than their exact functional forms. 

We analyze a stylized game in which competitive intermediaries offer menus of consumption bundles to agents, taking the principal's partial regulation of sold consumption goods as given. We say that a sub-game perfect equilibrium achieves \emph{intermediated implementation} if agents consume target consumption bundles and intermediaries break even on the equilibrium path. We mainly concern what policies achieve intermediated implementation in \emph{every} sub-game perfect equilibrium, and under what conditions. 

We consider two kinds of policies: \emph{per-unit fee schedules} and \emph{distribution regulations}. They constitute the weakest and strongest consumption regulations and are commonly used in practice. Their effectiveness hinges on whether intermediaries have \emph{interdependent} or \emph{private values} or, more specifically, whether their payoffs depend directly on agent hidden characteristics. In the private-value case, intermediated implementation can be achieved through a per-unit fee schedule that charges the target profit from selling every consumption good. In the interdependent-value case, per-unit fee schedules cannot generally be used to achieve intermediated implementation, whereas regulating the distribution of sold consumption goods can under weak regularity conditions. 

To illustrate, suppose, in the example of car sales, 20 percent of the customers have high tastes for cars and 80 percent have low tastes for cars. Dealers have interdependent values if sales profit depends directly on customer's taste, and they have private values otherwise. The manufacturer produces high- and low-quality cars and adopt two kinds of policies in practice: (1) \emph{invoice-price schedules} that charge possibly different per-unit fees for different car models, and (2) \emph{quota schemes} mandating that dealers sell a certain variety of cars. Through tracking the unique VIN number associated with each car and monitoring dealer inventories, the manufacturer can threaten to terminate business unless sales records are deemed satisfactory. A recent quote from a salesman at a Benz dealership in the Bay Area goes: ``We still have 2018 CLA-class left and they must go, otherwise the mixture doesn't look right and Benz won't deliver us new cars.''

In the private-value case, charging intermediaries the target profit from selling  each car achieves intermediated implementation. The reason is that under this per-unit fee schedule, a car makes profit if and only if it is sold above its target price. Meanwhile, incentive compatibility implies that each customer prefers his target car-price bundle to that of other customers, let alone bundles that are profitable to dealers. Thus the Bertrand competition between dealers drives profit to zero and sustains the target car-price rule in an equilibrium. Indeed, all equilibria are payoff equivalent for customers, and they all achieve intermediated implementation if no customer is indifferent under the target car-price rule. 

In the interdependent-value case, the adverse selection between customers and dealers could render invoice-price schedules ineffective. As will be shown later, undercutting the price of low-quality cars attracts both types of customers and makes a profit if high-taste customers are more profitable to serve (e.g., they drive less aggressively and incur lower repair costs under warranty) and have a binding incentive compatibility constraint under the target car-price rule. Consider instead a quota scheme mandating that among the cars sold by each dealer, 20 percent must be of high quality and 80 percent be of low quality. In order to meet this requirement, the only way that dealers can deviate from selling target-quality cars to customers is to sell high-quality cars to low-taste customers and low-quality cars to high-taste customers. However, since this permuted consumption rule is decreasing in  customer's taste, it cannot be part of any incentive compatible allocation and hence the result of dealer's profitable deviation. Thus in equilibrium, dealers offer target-quality cars to customers,  and competition drives prices to target levels. 

The above findings generalize to environments in which agents have quasi-linear utilities and multi-dimensional hidden characteristics. Based on Rochet (1987)'s characterization of incentive compatibility by cyclic monotonicity (CMON), we show that distribution regulation achieves intermediated implementation in every sub-game perfect equilibrium if and only if permuting consumption goods among agents yields a distinct total utility of consumption---a condition termed (\ref{eqn_du}). In single-dimensional environments, (\ref{eqn_du}) is satisfied if agent utilities satisfy the single-crossing property. In multi-dimensional environments, the profiles of agent utilities ruled out by (\ref{eqn_du}) are negligible compared to those that can be attained by incentive compatible allocations. 
  
We give two more applications of our results. We first examine how, under market-based healthcare systems, the government can achieve the target insurance provision through partial regulations of sold insurance policies. This application features interdependent values, as the profit from insurance sales depends directly on the patient's risk type. We find that when patients have CARA utilities, per-unit coverage-plan subsidies cannot generally achieve the government's target, whereas regulating the variety of sold coverage plans can. This finding justifies recent rulings by the Affordable Care Act, which mandate that all participating companies in the health insurance exchange sell a variety of coverage plans and penalize excessive sales of low-coverage plans (\cite{folger}; \cite{cox}).

We next examine income redistribution among workers with hidden abilities, recognizing in that in reality, it is often firms rather than the government who specify and monitor worker performance based on nonverifiable or even subjective information (\cite{levin}). We distinguish between whether firms can contract on effective labor outputs or labor hours, which correspond to the private-value case and interdependent-value case, respectively. We show that the Mirrleesian income tax schedule achieves the target redistribution in the first case, but must be used together with distribution regulations such as disability quota in the second case. 

\subsection{Related Literature}\label{sec_literature}
\paragraph{Adverse selection}
 Since \cite{rothschildstiglitz}, economists have long recognized that interdependent values, or adverse selection, may cause the nonexistence or inefficiency of competitive equilibria. Among the various attempts to tackle this challenge, \cite{prescotttownsend}, \cite{gale}, \cite{dubey}, \cite{bisingottardi} and \cite{azevedo} examine Walrasian equilibria in which players take the price and trader composition of any potential contract as given;  \cite{miyazaki}, \cite{wilson} and \cite{riley} develop solution concepts that allow players to anticipate the reactions of other players to their contractual offers; last but not least, \cite{inderst} and \cite{search} demonstrate the usefulness of capacity constraints for restoring equilibrium existence.\footnote{\cite{attar} use latent contracts to deter cream-skimming deviations when competition is nonexclusive. Here competition is exclusive.}

The current work differs from the above ones in four aspects. First, we implement any IC-IR social choice rule, which differs generally from the equilibrium outcomes of laissez-faire economies. Second, we study a different game in which intermediaries can specify all facets of the consumption bundle and face no limit on the menu size on the one hand but are subject to the principal's partial regulation on the other hand. Third, we adopt a standard solution concept.\footnote{\cite{hellwig}, \cite{scheuer} and the references therein provide game-theoretic foundations for the solution concepts developed by \cite{miyazaki}, \cite{wilson} and \cite{riley} through introducing additional mechanisms (e.g., withdraw, renegotiation) into the original game of \cite{rothschildstiglitz}---a trick we do not employ here.} Fourth, we give necessary and sufficient conditions for achieving intermediated implementation in every equilibrium based on model primitives rather than obtaining equilibrium uniqueness through refinements. The latter, if performed, could only relax our conditions. 

\paragraph{Vertical coordination} We make two additions to the literature on vertical coordination. First, we demonstrate the implementability of any IC-IR social choice rule through intermediaries.
Second, we investigate a new kind of vertical externality stemming from the adverse selection between agents and intermediaries. Adverse selection is present in important vertical structures such as retail sales, healthcare regulation and income redistribution.  The resulting vertical externality is shown to be solvable by distribution regulations, e.g., full-line forcing that mandates the sale of all consumption goods. \cite{shaffer} and \cite{verge} study vertical control problems featuring a multi-product monopolistic intermediary. Full-line forcing is shown to resolve the vertical externality stemming from double marginalization and the imperfect substitutability between products. See \cite{tirolebook} and \cite{reyverge} for literature surveys.


The remainder of this paper proceeds as follows: Section \ref{sec_model} introduces the model; Section \ref{sec_analysis} conducts analysis; Section \ref{sec_extension} investigates extensions; Section \ref{sec_application} gives applications; Section \ref{sec_conclusion} concludes. Omitted materials and proofs can be found in the appendices. 

\section{Model}\label{sec_model}
\subsection{Setup}\label{sec_setup}
\paragraph{Primitives} The economy consists of a principal, finite $I$ intermediaries and a unit mass of infinitesimal agents. Agents draw hidden characteristics (denoted by $\theta$) independently from a finite set $\Theta$ according to a probability function $P_{\theta}$. A consumption bundle $\left(x,y\right)$ consists of a consumption good $x \in X \subset \mathbb{R}^d$ and a price $y \in Y=\mathbb{R}$. Agents have unit demand for consumption bundles and, for now, quasi-linear utilities $u\left(x,y,\theta\right)=v\left(x,\theta\right)-y$. A constant-returns-to-scale technology yields a profit $\pi\left(x,y,\theta\right)$ from serving a bundle $\left(x,y\right)$ to a type $\theta$ agent, where the function $\pi$ is continuous and strictly increasing in $y$. Let $\bf{0}$ denote the null bundle that yields zero reservation utility to inactive players, and assume without loss of generality that ${\bf{0}} \in X \times Y$.

\paragraph{Target social choice rule} Let $\left(\hat{x},\hat{y}\right): \Theta \rightarrow X \times Y$ be \emph{any} social choice rule that is incentive compatible and individually rational (IC-IR) for agents: 
\begin{align*}
\tag{IC$_{\theta}$} & u\left(\hat{x}\left(\theta \right), \hat{y}\left(\theta\right),\theta\right) \geq u\left(\hat{x}\left(\theta'\right), \hat{y}\left(\theta'\right),\theta\right) \text{ } \forall \theta, \theta', \label{eqn_ic}\\
\tag{IR$_{\theta}$}&\text{and } u\left(\hat{x}\left(\theta\right), \hat{y}\left(\theta\right),\theta\right) \geq 0 \text{ } \forall \theta.\label{eqn_ir}
\end{align*}
In what follows, we will take $\left(\hat{x}, \hat{y}\right)$ as given and focus on its implementation.\footnote{We remain agnostic about the principal's objective for the most part and refer the reader to Sections \ref{sec_sales} and \ref{sec_application} for concrete examples.} In the textbook case where the principal owns intermediaries, she can simply offer the menu of target consumption bundles $\left\{\left(\hat{x}\left(\theta\right), \hat{y}\left(\theta\right)\right): \theta \in \Theta \right\}$ to agents and let them self-select; after that, intermediaries passively serve the selected bundles to agents and surrender profits to the principal. 

\paragraph{Intermediated implementation}  Now suppose the technology for serving agents is owned by intermediaries, whereas the principal is limited to regulating sold consumption goods by practical barriers. In what follows, we will take these barriers as given and examine when and how the principal can implement the target social choice rule through intermediaries. We focus on the case of \emph{competitive intermediaries} $I \geq 2$ and relegate the case of \emph{monopolistic intermediary} $I=1$ to Appendix \ref{sec_monopolistic}.

Formally, for each $i=1,\cdots, I$, let $\mu_i$ denote the measure on $X\times Y \times \Theta$ induced by intermediary $i$'s sold bundles and $\nu_i$ denote the measure on $X$ induced by $\mu_i$. A policy $\psi: \Delta\left(X\right) \rightarrow \mathbb{R}$ maps measures on sold consumption goods to the reals. Based on $\nu_i$, the principal charges intermediary $i$ a fee $\psi\left(\nu_i\right)$, leaving the latter with the following net profit: 
\[
\int_{\left(x,y,\theta\right)} \pi\left(x,y,\theta\right) d\mu_i - \psi\left(\nu_i\right).
\]
Time evolves as follows: 

\begin{enumerate}
\item the principal commits to a policy $\psi: \Delta\left(X\right)\rightarrow \mathbb{R}$;
\item intermediaries propose menus of \emph{deterministic} consumption bundles $\mathcal{M}_i \in 2^{X \times Y}$, $i=1,\cdots, I$ to agents;\footnote{If intermediaries can offer lotteries of consumption goods to agents, then let the principal regulate these lotteries, i.e., $\psi:\Delta\left(\Delta\left(X\right)\right) \rightarrow \mathbb{R}$, and all upcoming results will remain valid.}
\item each agent selects a bundle from $\cup_i \mathcal{M}_{i} \cup \left\{{\bf{0}}\right\}$; 
\item intermediaries pay fees $\psi\left(\nu_i\right)$, $i=1,\cdots, I$ to the principal.
\end{enumerate}

Our solution concept is \emph{sub-game perfect equilibrium} (hereafter  equilibrium). Below is our notion of implementation; we will consider \emph{both} partial and full implementations: 

\begin{defn}\label{defn_ii}
A sub-game perfect equilibrium induced by a policy $\psi$ \emph{achieves intermediated implementation} if all agents consume  target consumption bundles and all intermediaries break even on the equilibrium path.
\end{defn}

\paragraph{Notations and assumption} Let $\hat{x}\left(\Theta\right)$, $\hat{y}\left(\Theta\right)$ and $\left(\hat{x}, \hat{y}\right)\left(\Theta\right)$ denote the images of $\Theta$ under the mappings $\hat{x}$, $\hat{y}$ and $\left(\hat{x},\hat{y}\right)$, respectively. Assume throughout that $\hat{x}\left(\Theta\right)=X$ and see Sections \ref{sec_sales} and \ref{sec_application} for interpretations in concrete examples. For each $x \in \hat{x}\left(\Theta\right)$, $\hat{y}\left(x\right)$ is the unique price $y \in \hat{y}\left(\Theta\right)$ such that $\left(x,y\right) \in \left(\hat{x}, \hat{y}\right)\left(\Theta\right)$, and 
\[
\hat{\pi}\left(x\right)=\mathbb{E}_{\theta}\left[\pi\left(x, \hat{y}\left(x\right),\theta\right)\mid \left(\hat{x}, \hat{y}\right)\left(\theta
\right)=\left(x,\hat{y}\left(x\right)\right)\right]
\]
is the expected profit from serving the bundle $\left(x, \hat{y}\left(x\right)\right)$ to its target agents.

\subsection{Candidate Policies}\label{sec_candidatepolicy}
We consider two kinds of policies: \emph{per-unit fee schedules} and \emph{distribution regulations}.  

A per-unit fee schedule, denoted by $\psi_{per-unit}$,   charges intermediaries a fee $t\left(x\right)$ for selling every unit of  consumption good $x \in X$. The fees can differ across $x$'s, and the total charge is additive across the sold units. The focus will be on $\hat{\psi}_{per-unit}$, which charges $t\left(x\right)=\hat{\pi}\left(x\right)$ and makes intermediaries break even, bundle by bundle, under the target social choice rule. The total charge to intermediary $i$ is 
\[
\hat{\psi}_{per-unit}\left(\nu_i\right)=\int_{x \in X} \hat{\pi}\left(x\right)  d\nu_i
\]
and, if negative, constitutes a subsidy.

Distribution regulations implement $\hat{\psi}_{per-unit}$ if the measure $\nu_i$ on sold consumption goods satisfies certain properties but otherwise inflict a severe penalty (denoted by ``$+\infty$'') on  intermediary $i$. The focus will be on $\hat{\psi}_{distr}$, which mandates that $\nu_i$ match the target consumption distribution $\hat{P}_x$: 
\[
\hat{\psi}_{distr}\left(\nu_i\right)=
\begin{cases}
\hat{\psi}_{per-unit}\left(\nu_i\right) &\text{ if } \frac{\nu_i}{\int_{x \in X} d\nu_i}=\hat{P}_x,\\
+\infty &\text{ otherwise. }
\end{cases}
\]

\paragraph{Discussions} The penalty term ``$+\infty$''  deters intermediaries from missing the distributional target and is never invoked on the equilibrium path. Its interpretation can be context specific, ranging from unmodeled losses from relationship severance in the car-sales example to legal penalties in the applications given in Section \ref{sec_application}.

Setting the transfers in the candidate policies to $\int_{x\in X} \hat{\pi}\left(x\right)d\nu_i$ is required by the definition of intermediated implementation, according to which intermediaries must surrender $\int_{x\in X} \hat{\pi}\left(x\right)d\nu_i$ to the principal on the equilibrium path.  $\hat{\psi}_{per-unit}$ is not the only per-unit fee schedule that implements this transfer rule; other per-unit fee schedules that might as well work would entail cross subsidization and make intermediaries break even on average.\footnote{Restricting attention to per-unit fee schedules is without loss due to the assumption of constant returns to scale.}  When demonstrating the effectiveness of the candidate policies in Theorems \ref{thm_private} and \ref{thm_interdependent}, we use $\hat{\psi}_{per-unit}$ to implement the on-path transfer rule because doing so  turns out to be without loss of generality.\footnote{As the reader will soon realize, replacing $\hat{\psi}_{per-unit}$ with per-unit fee schedules that entail cross subsidization invalidates Theorem \ref{thm_private} because intermediaries won't serve bundles that incur losses in the private-value case. It doesn't affect Theorem \ref{thm_interdependent} because given that distribution regulation already implements the target consumption rule, all that is left is to charge the target profit from intermediaries and make them break even on average. } When proving the negative result in Section \ref{sec_interdependent_perunit}, we consider all per-unit fee schedules, showing that none of them achieves intermediated implementation in any equilibrium.   

Among all consumption regulations that use $\hat{\psi}_{per-unit}$ to implement the on-path transfer rule, $\hat{\psi}_{per-unit}$ is clearly the \emph{weakest} policy. $\hat{\psi}_{distr}$ constitutes a \emph{strongest} policy because it eliminates all detectable deviations that intermediaries can possibly commit.

\paragraph{Key condition} The effectiveness of candidate policies hinges on whether intermediaries have private or interdependent values:

\begin{defn}\label{defn_pv}
Intermediaries have \emph{private values} if $\pi\left(x,y,\theta\right)$ is independent of $\theta$ for any $(x,y)\in X \times Y$, and they have \emph{interdependent values} otherwise.
\end{defn}

\subsection{Illustrative Example}\label{sec_sales}
In the car-sales example laid out in Section \ref{sec_intro}, suppose  customers have unit demand for cars. A customer's taste, denoted by $\theta$, is his private information, and his reservation utility is normalized to zero. 
A consumption bundle $\left(x,y\right)$ is a pair of car model $x \in X=\left\{\text{C-class}, \text{E-class}, \ldots \right\}$ and price $y \in Y= \mathbb{R}$. It yields a utility $u\left(x,y,\theta\right)=v\left(x,\theta\right)-y$ to type $\theta$ customers and a sales profit $\pi\left(x,y,\theta\right)=y - c\left(x,\theta\right)$ to dealers.

The term $c\left(x,\theta\right)$ represents the distribution cost incurred on advertisement, staffing, repair under warranty, etc.. In the case where $c\left(x, \theta\right)$ depends on $\theta$ (e.g., low-taste customers drive more aggressively and incur higher repair costs under warranty than high-taste customers), dealers have interdependent values. Otherwise dealers have private values.

Consider first the textbook case in which the manufacturer owns dealers and directly sells cars to customers. The target social choice rule $\left(\hat{x}, \hat{y}\right):\Theta \rightarrow X \times Y$ solves: 
\[\max_{\left(x,y\right):\Theta \rightarrow X \times Y}\int_{\theta \in \Theta} \pi\left(x\left(\theta\right),y\left(\theta\right),\theta\right)dP_{\theta}+d\left(\left(x,y\right)\right) \text{ s.t. (\ref{eqn_ic}), (\ref{eqn_ir}) } \forall \theta, \]
where the term $d\left(\left(x,y\right)\right)$ captures other concerns than sales profit (e.g., the wedge introduced by warranty, brand loyalty, etc.) and can depend on the entire social choice rule $\left(x,y\right)$. Since the manufacturer controls the car supply, it is without loss to assume $X=\hat{x}\left(\Theta\right)$. 

In reality, the distribution technology is owned by dealers, whereas the manufacturer is prohibited from influencing car prices, directly or indirectly, by the legal barrier discussed in Section \ref{sec_intro}. Yet she still controls the car supply and tracks sold cars through their unique VIN's and therefore wields policies of form $\psi: \Delta\left(X\right) \rightarrow \mathbb{R}$. 

When trying to implement the target car-price rule, the manufacturer faces two challenges. First, due to the potential difference in the objective function,  dealers may wish to charge off-target prices, e.g., in order to cover the excessive repair costs under warranty. Or they may wish to sell off-target varieties of cars, e.g., ignoring that cultivating brand loyalty requires serving low-end models to young, first-time customers. 

The second challenge stems from customer's hidden taste, which creates adverse selection between customers and dealers in the interdependent-value case. As shown in Section \ref{sec_interdependent_perunit}, the resulting distortion could persist even if the principal's goal is to maximize expected sales profit, i.e., $d\left(\left(x,y\right)\right) \equiv 0$. 

We examine the effectiveness of candidate policies in overcoming these challenges. In reality,
\begin{itemize}
\item $\psi_{per-unit}$ represents invoice-price schedules that charge possibly different per-unit fees for different car models; $\hat{\psi}_{per-unit}$ is a special case that makes dealers break even, bundle by bundle, under the target car-price rule;
\item $\hat{\psi}_{distr}$ mandates that dealers sell the target variety of cars and threatens to terminate business if the distributional target is missed. 
\end{itemize}

 \section{Main Results}\label{sec_analysis}
 
\subsection{Private Values}\label{sec_private}
The next theorem demonstrates the effectiveness of $\hat{\psi}_{per-unit}$ in the private-value case: Part (i) of it shows that $\hat{\psi}_{per-unit}$ achieves intermediated implementation in an equilibrium; Part (ii) prescribes the necessary and sufficient condition for $\hat{\psi}_{per-unit}$ to achieve intermediated implementation in every equilibrium: 

\begin{thm}\label{thm_private}
Fix any IC-IR social choice rule $\left(\hat{x}, \hat{y}\right): \Theta \rightarrow X \times Y$ and the  corresponding per-unit fee schedule $\hat{\psi}_{per-unit}$. In the private-value case, 
\begin{enumerate}[(i)]
\item there exists a sub-game perfect equilibrium that achieves intermediated implementation, in which $\mathcal{M}_i^*=\left\{\left(\hat{x}\left(\theta\right), \hat{y}\left(\theta\right)\right): \theta \in \Theta \right\}$ for $i=1,\cdots, I$ and each type $\theta$ of agent consumes $\left(\hat{x}\left(\theta\right), \hat{y}\left(\theta\right)\right)$ on the equilibrium path;
\item every sub-game perfect equilibrium achieves intermediated implementation if and only if no agent is indifferent between his bundle and any other bundle under the target social choice rule.  
\end{enumerate}
\end{thm} 

To develop intuition for Theorem \ref{thm_private}(i), recall that $\hat{\psi}_{per-unit}$ makes intermediaries break even, bundle by bundle, under the target social choice rule. In the private-value case, this means that a consumption good makes a profit if and only if it is sold above its target price. In an augmented laissez-faire economy where the target social choice rule constitutes the production frontier, incentive compatibility implies that each agent most prefers his target consumption bundle among all bundles located on the production frontier. Thus the target social choice rule is Pareto efficient and can be sustained by the Bertrand competition between intermediaries in an equilibrium.\footnote{\cite{fagart} and \cite{bejt} examine the efficiency property of laissez-faire economies featuring private values. We devise policy interventions that implement any IC-IR social choice rule through intermediaries. }

The proof of Theorem \ref{thm_private}(ii) exploits the next lemma: 

\begin{lem}\label{lem_payoffequivalence}
Let everything be as in Theorem \ref{thm_private}. Then all agents obtain target levels of utilities and pay target prices for consumption goods in every sub-game perfect equilibrium induced by $\hat{\psi}_{per-unit}$.
\end{lem}

According to Lemma \ref{lem_payoffequivalence}, if an equilibrium induced by $\hat{\psi}_{per-unit}$ fails to achieve intermediated implementation, then the only explanation is that some agent opts into a bundle he is indifferent about under the target social choice rule. Eliminating such indifference leads to the implementation of the target social choice rule in every equilibrium.

\subsection{Interdependent Values}\label{sec_interdependent}

\subsubsection{Per-unit Fee Schedule}\label{sec_interdependent_perunit}
In this section, we illustrate through a counterexample that per-unit fee schedules cannot generally achieve intermediated implementation in the interdependent-value case.

\begin{example}[label=exm_interdependent]
Suppose $\Theta=\left\{\theta_1, \theta_2\right\}$ and $\pi\left(x,y,\theta_2\right)>\pi\left(x,y,\theta_1\right)$ for any $\left(x,y\right)$, i.e., intermediaries prefer to serve type $\theta_2$ agents than type $\theta_1$ agents, other things being equal. Fix any target social choice rule with a binding (IC$_{\theta_2}$) constraint and a slack (IC$_{\theta_1}$) constraint. Below we argue in two steps that no per-unit fee schedule achieves intermediated implementation in any equilibrium. 

Consider first $\hat{\psi}_{per-unit}$.  Let $\left(\hat{x}_i, \hat{y}_i\right)$ denote the target consumption bundle of type $\theta_i$ agent, $i=1,2$, and notice that 
\[\pi\left(\hat{x}_1, \hat{y}_1, \theta_2\right)>\hat{\psi}_{per-unit}\left(\hat{x}_1\right) = \pi\left(\hat{x}_1, \hat{y}_1, \theta_1\right).\]
Suppose, to the contrary, that an equilibrium induced by $\hat{\psi}_{per-unit}$ achieves intermediated implementation. Consider a deviation by $i$ that adds $\left(\hat{x}_1, \hat{y}_1-\epsilon\right)$ to its menu, where $\epsilon$ is a small positive number. A close inspection reveals that the new consumption bundle is preferred by both types of agents to their target consumption bundles, and it strictly increases intermediary $i$'s profit:
\begin{align*}
&\rho \pi\left(\hat{x}_1, \hat{y}_1,\theta_2\right)+(1-\rho) \pi\left(\hat{x}_1, \hat{y}_1,\theta_1\right) - \hat{\psi}_{per-unit}\left(\hat{x}_1\right)\\
&=\rho \left[\pi\left(\hat{x}_1, \hat{y}_1,\theta_2\right)-\pi\left(\hat{x}_1, \hat{y}_1,\theta_1\right)\right]>0,
\end{align*} 
where $\rho$ denotes the population of type $\theta_2$ agents in the economy. This leads to a contradiction. 

Since $\hat{\psi}_{per-unit}$ makes intermediaries break even, bundle by bundle, under the target social choice rule, it remains open whether intermediated implementation could be achieved through cross subsidization. As shown in Appendix \ref{sec_a1}, the answer to this question is still negative.

Here are two remarks. First, there is no inconsistency between our negative result and the positive result of  \cite{rothschildstiglitz} showing that separating equilibria sometimes exist. The main reason is that in those separating equilibria, agent incentive constraints bind in different directions from ours. 

Second, the above distortion could persist even if the principal's goal is to maximize expected sales profit. Aside from interdependent values, all we need is that the target social choice rule has binding  (IC$_{\theta_2}$) constraint and a slack (IC$_{\theta_1}$) constraint, which is well known to be the case under mild regularity conditions  (e.g., $\Theta, X \subset \mathbb{R}$, $\theta_1<\theta_2$ and agent utility has strict increasing differences in $\left(x,\theta\right)$).
\end{example} 

\subsubsection{Distribution Regulation}\label{sec_interdependent_distribution}
In this section, we first demonstrate that $\hat{\psi}_{distr}$ achieves intermediated implementation in an equilibrium, and then give the necessary and sufficient condition for $\hat{\psi}_{distr}$ to achieve intermediated implementation in every equilibrium.

A few definitions before we go into detail. The next two definitions are standard, and Definition \ref{defn_implementable} should not be confused with intermediated implementation:

\begin{defn}\label{defn_implementable}
A consumption rule $x: \Theta \rightarrow X$ is \emph{implementable} if there exists a transfer rule $y:\Theta \rightarrow Y$ such that the social choice rule $\left(x,y\right): \Theta \rightarrow X\times Y$ is incentive compatible for all agents.
\end{defn}

 \begin{defn}\label{defn_cycle}
A bijection $\sigma: \Theta \rightarrow \Theta$ constitutes a \emph{cyclic permutation of $\Theta$} if $\theta \rightarrow \sigma \left(\theta\right) \rightarrow \sigma \circ \sigma\left(\theta\right) \rightarrow \cdots$ forms a $|\Theta|$-cycle.
\end{defn}

The next definition says that a consumption rule satisfies (\ref{eqn_du}) if permuting consumption goods among agents yields a distinct total utility of consumption:

\begin{defn}\label{defn_du}
A consumption rule $x: \Theta \rightarrow X$ satisfies \emph{(\ref{eqn_du})} if the following holds for any $\Theta'\subset \Theta$ such that $x(\theta) \neq x(\theta')$ $\forall \theta, \theta' \in \Theta'$ and any cyclic permutation $\sigma: \Theta' \rightarrow \Theta'$:
\begin{equation}\label{eqn_du}
\tag{DU} \sum_{\theta \in \Theta'} v(x(\theta), \theta) \neq \sum_{\theta \in \Theta
'} v(x(\sigma(\theta)), \theta). 
\end{equation}
\end{defn}

We now state the main result of this section: 
\begin{thm}\label{thm_interdependent}
Fix any IC-IR social choice rule $\left(\hat{x}, \hat{y}\right): \Theta \rightarrow X \times Y$ and the corresponding distribution regulation $\hat{\psi}_{distr}$. Then, 
\begin{enumerate}[(i)]
\item there exists a sub-game perfect equilibrium that achieves intermediated implementation, in which $\mathcal{M}_i^*=\left\{\left(\hat{x}\left(\theta\right), \hat{y}\left(\theta\right)\right): \theta \in \Theta\right\}$ for $i=1,\cdots, I$ and each type $\theta$ of agent consumes $\left(\hat{x}\left(\theta\right), \hat{y}\left(\theta\right)\right)$ on the equilibrium path;
\item every sub-game perfect equilibrium achieves intermediated implementation if and only if $\hat{x}: \Theta \rightarrow X$ satisfies (\ref{eqn_du}).
\end{enumerate}
\end{thm}

\paragraph{Proof sketch} Below we sketch the proof of Theorem \ref{thm_interdependent}. The formal analysis is relegated to Appendix \ref{sec_a2}. 

\paragraph{Part (i)} Under $\hat{\psi}_{distr}$, the only way that intermediaries can deviate from serving target consumption bundles to agents and make profits is to permute consumption goods among agents. Thus, if the permuted consumption rule cannot be the result of intermediary's profitable deviation (as illustrated in the next example), then the target consumption rule gets implemented, and the competition between intermediaries drives prices to target levels:

\begin{example}[label=exm:binary]
Suppose $\Theta=\left\{\theta_1,\theta_2\right\}$ and $\hat{x}\left(\theta_1\right) \neq \hat{x}\left(\theta_2\right)$. Write $\hat{x}_i=\hat{x}\left(\theta_i\right)$ for $i=1,2$ and $v_{ij}=v\left(\hat{x}_i, \theta_j\right)$ for $i,j=1,2$. Recall Rochet (1987)'s characterization of  implementable consumption rules by cyclic monotonicity (CMON). In the current setting, applying (CMON) to the target consumption rule yields 
\[v_{11}+v_{22} \geq v_{12}+v_{21}.\]

Consider two cases. First, if the target consumption rule satisfies (\ref{eqn_du}), i.e., \[v_{11}+v_{22} \neq v_{12}+v_{21},\] then the permuted consumption rule that assigns $\hat{x}_1$ to type $\theta_2$ agents and $\hat{x}_2$ to type $\theta_1$ violates (CMON): \[v_{12}+v_{21} <v_{11}+v_{22},\]
and therefore cannot be part of any incentive compatible allocation or  the result of intermediary's profitable deviation.  

Second, if the target consumption rule violates (\ref{eqn_du}), then agents must be indifferent under all incentive compatible allocations formed based on $\hat{x}_1$ and $\hat{x}_2$. This can be seen from writing out the incentive compatibility constraints:
 \[v_{11}-y_1 \geq v_{21}-y_2 \text{ and } v_{22}-y_2 \geq v_{12}-y_1,\]
and noting that both inequalities are binding in the absence of (\ref{eqn_du}), i.e.,
\begin{equation}\label{e1}
v_{21}-v_{11}=v_{22}-v_{12}=y_2-y_1.
\end{equation}
If agents break ties in favor of their target consumption goods off the equilibrium path, then the permuted consumption rule cannot arise off the equilibrium path, let alone be the result of intermediary's profitable deviation.  
\end{example} 

\paragraph{Part (ii)}  The ``if'' direction is immediate. To prove the ``only if'' direction, recall that absent (\ref{eqn_du}), agents must indifferent between each other's consumption bundle under all incentive compatible allocations formed based on the target consumption goods. Based on this observation, we construct, in Appendix \ref{sec_a2}, ``bad equilibria'' in which indifferent agents swap  consumption goods among themselves on the equilibrium path. Contrary to the private-value case in which all equilibria induced by $\hat{\psi}_{per-unit}$ (hence $\hat{\psi}_{distr}$) are payoff and price equivalent (Lemma \ref{lem_payoffequivalence}), here we could lose these properties due to again interdependent values. The next example illustrates this idea.

\begin{example}[continues=exm:binary]
The target social choice rule violates (\ref{eqn_du}), $\pi\left(x,y,\theta\right)=y-c\left(x,\theta\right)$ and each type of agent constitutes half the population. Write $c_{ij}=c\left(\hat{x}_i, \theta_j\right)$ for $i, j=1,2$, and normalize $\pi\left(\hat{x}_i, \hat{y}_i, \theta_i\right) \coloneqq \hat{\pi}\left(x_i\right)$ to zero for $i=1,2$. 

Absent (\ref{eqn_du}), agents are indifferent under all incentive compatible allocations formed based on $\hat{x}_1$ and $\hat{x}_2$. Take any incentive compatible allocation that assigns $\left(\hat{x}_2, y_2\right)$ to type $\theta_1$ agents and $\left(\hat{x}_1, y_1\right)$ to type $\theta_2$ agents, and note that it must satisfy Equation (\ref{e1}). To sustain this allocation in equilibrium, we must satisfy intermediary's zero-profit condition: 
\begin{equation}\label{e2}
y_1+y_2=c_{12}+c_{21}, 
\end{equation}
Solving Equations (\ref{e1}) and (\ref{e2}) yields a unique solution $\left(y_1^{\ast}, y_2^{\ast}\right)$.

Suppose an intermediary deviates unilaterally from offering the above allocation to agents. To make a profit, the deviator must offer $\left(\hat{x}_1, y_1'\right)$ to type $\theta_1$ agents and $\left(\hat{x}_2, y_2'\right)$ to type $\theta_2$ agents. The new allocation must be incentive compatible and therefore satisfy Equation (\ref{e1}), and it must incur no loss to the deviator: 
\begin{equation}\label{e3}
y_1'+y_2'\geq c_{11}+c_{22}.
\end{equation}

In the case where $c_{11}+c_{22} \gg c_{21}+ c_{12}$ (which happens only in the interdependent-value case), each type $\theta_i$ of agent  strictly prefers $\left(\hat{x}_{-i}, y_{-i}^*\right)$ to $\left(\hat{x}_i, y_i'\right)$,  where $\left(y_1', y_2'\right)$ satisfies Conditions (\ref{e1}) and  (\ref{e3}). Thus the above deviation is unprofitable, and the initial  allocation can be sustained in an equilibrium.
\end{example}

\paragraph{Discussions} The condition (\ref{eqn_du}) can be easily satisfied in both single- and multi-dimensional environments:

\begin{example}\label{exm_du1}
If $X, \Theta \subset \mathbb{R}$ and $v\left(x,\theta\right)$ has strict increasing differences in $\left(x,\theta\right)$, then any implementable consumption rule is nondecreasing in $\theta$ (\cite{milgromshannon}) and therefore must satisfy (\ref{eqn_du}).
\end{example}

\begin{example}\label{exm_du2}
Write $\Theta=\left\{\theta_1,\cdots, \theta_N\right\}$ and a consumption rule as $\left(x_1,\cdots, x_N\right)$, where $x_i$ stands for the consumption good of type $\theta_i$ agent. Let $V$ be the $N \times N$ matrix whose $ij^{th}$ entry is $v\left(x_i, \theta_j\right)$. Let $\mathbb{I}$ denote the diagonal matrix and $\Sigma$ a typical permutation matrix, both of size $N \times N$. By \cite{rochet}, the utility matrices in 
$\displaystyle \bigcap_{\Sigma \neq \mathbb{I}}\left\{V:\left(\mathbb{I}-\Sigma\right) \cdot V \geq 0\right\}$
can arise under incentive compatible allocations, whereas the utility matrices ruled out by (\ref{eqn_du}) belong to
$\displaystyle \bigcup_{\Sigma \neq \mathbb{I}} \left\{V: \left(\mathbb{I}-\Sigma\right) \cdot V=0 \right\}$. The first set has dimension $N \times N$ whereas the second one has dimension less than $N \times N$.
\end{example} 

\section{Extensions}\label{sec_extension}
\paragraph{Unknown target distribution} In the case where the principal doesn't know the exact target consumption distribution, there are at least three things she can do.  First, she can adjust the target social choice rule based on what she knows, making sure it is incentive compatible and individually rational for agents. Second, she can ask intermediaries to match distributions among themselves without reference to the target distribution:\footnote{\cite{segal} examines an optimal pricing problem in which the principal sets the price charged to agent $i$ based on the demand curve constructed from agent $-i$'s reports.  } 
\[\tilde{\psi}_{distr}\left(\nu_i\right)=\begin{cases}
\hat{\psi}_{per-unit}\left(\nu_i\right) & \text{ if } \frac{\nu_i}{\int_{x \in X} d\nu_i}= \frac{\nu_j}{\int_{x \in X} d\nu_j} \text{ } \forall j \neq i,\\
+\infty & \text{ otherwise}.
\end{cases}\]
Third, she can use distribution regulations that exploit less information than the full consumption distribution (hereafter \emph{coarse} distribution regulations). An example is full-line forcing, which mandates the sale of all consumption goods:\footnote{$\supp\left(\cdot\right)$ denotes the support of a measure.}
\[
\breve{\psi}_{distr}\left(\nu_i\right)=\begin{cases}
\hat{\psi}_{per-unit}\left(\nu_i\right) & \text{ if } \supp\left(\nu_i\right)=X,\\
+\infty & \text{ otherwise},
\end{cases}
\]
 The next corollary demonstrates the effectiveness of $\tilde{\psi}_{distr}$ and $\breve{\psi}_{distr}$ in an equilibrium:\footnote{In general, one should not expect to achieve full implementation using coarser distribution regulations than $\hat{\psi}_{distr}$, because they cannot prevent chains of exchanges (e.g.,  some type $\theta_1$ agents consume $\hat{x}_2$, all type $\theta_2$ agents consume $x_2$) from arising in equilibrium. }
\begin{cor}\label{cor_coarse}
For any IC-IR social choice rule $\left(\hat{x}, \hat{y}\right): \Theta \rightarrow X \times Y$, Theorem \ref{thm_interdependent}(i) holds for the corresponding $\tilde{\psi}_{distr}$ and $\breve{\psi}_{distr}$. 
\end{cor}

\begin{proof}
Assuming that agents of the same type adopt the same off-path tie-breaking rule in the proof of Theorem \ref{thm_interdependent}(i) yields the effectiveness of $\breve{\psi}_{distr}$. The effectiveness of $\tilde{\psi}_{distr}$ needs no more proof.
\end{proof}

\paragraph{Competing mechanism} The following happen if intermediaries could offer competing mechanisms to agents.\footnote{Existing studies on competing mechanisms include \cite{epsteinpeters}, \cite{yamashita} and \cite{petersszentes} and are surveyed by \cite{peterssurvey}.} First, Theorems \ref{thm_private}(i) and \ref{thm_interdependent}(i) on partial implementation remain valid, so from the principal's perspective, the best achievable outcome through consumption regulation remains the target social choice rule. Second, in the presence of (\ref{eqn_du}), the same argument for Theorem \ref{thm_interdependent} shows that every equilibrium induced by $\hat{
\psi}_{distr}$ implements the target consumption rule. So far intermediaries have been competing \`{a} la Bertrand on the $y$-dimension. By using competing mechanisms, they could sustain more collusive levels of $y$'s, with the most collusive one solving:
\[
\max_{\left(\hat{x},y\right): \Theta \rightarrow X \times Y} \int_{\theta \in \Theta} \pi\left(\hat{x}\left(\theta\right),y\left(\theta\right),\theta\right) dP_{\theta} \text{ } \text{ s.t. (IC$_{\theta}$), (IR$_\theta$) } \forall \theta.
\]

\section{Other Applications}\label{sec_application}

\subsection{Market-based Healthcare Regulation}\label{sec_insurance}
There is a government, a finite number of insurance companies and a continuum of patients.  A patient's endowment is a finite random variable that equals $e_s$ with probability $\theta_s$, $s=1,\cdots, S$. The probability vector $\bm{\theta}=\left(\theta_1,\cdots, \theta_S\right)$ captures the patient's risk type and is his private information. The endowment realization is public information. 

An insurance policy $\left(\bm{x}, y\right)$ is a pair of consumption plan $\bm{x}=\left(x_1,\cdots, x_S\right) \in X=\left\{\text{gold plan}, \text{platinum plan}, \cdots\right\} \subset \mathbb{R}^S$ and premium $y \in Y =\mathbb{R}$. It yields an expected utility $u\left(\bm{x}, y, \bm{\theta}\right)=\sum_s \theta_s v\left(x_s-y\right)$ to type $\bm{\theta}$ patients and an expected profit $\pi\left(\bm{x}, y, \bm{\theta}\right)=y-\sum_s \theta_s \left(x_s-e_s\right)$ to insurance companies. Each patient can purchase at most one policy. Insurance companies have interdependent values.

The government wishes to implement an IC-IR social choice rule $\left(\hat{\bm{x}}, \hat{y}\right)$. Under market-based systems, she doesn't provide insurance to patients herself (in contrast to what happens under single-payer systems) but only partially regulate the policies sold by insurance companies. An example is coverage-plan subsidy, which regulates the consumption plan $\bm{x}$ but not the premium $y$. 

We examine the effectiveness of coverage-plan regulations in achieving intermediated implementation. The following are noteworthy: 
\begin{itemize}
\item the Affordable Care Act imposes strict rules on salable coverage plans and effectively limits $X$ to $\hat{\bm{x}}\left(\Theta\right)$;

\item $\psi_{per-unit}$ represents coverage-plan subsidies;

\item $\hat{\psi}_{distr}$ resembles recent rulings by the Affordable Care Act, which mandate that all participating companies in the health insurance exchange sell the target variety of coverage plans and penalize excessive sales of low-coverage plans (\cite{folger}; \cite{cox}). 
\end{itemize}

For the same reason as given in Example \ref{exm_interdependent}, the government cannot use coverage-plan subsidies to achieve the target insurance provision in general. The next proposition demonstrates the effectiveness of $\hat{\psi}_{distr}$ when patients have CARA utilities:  

\begin{cor}\label{cor_insurance}
Suppose $v\left(c\right)=1-\exp\left(-\lambda c\right)$ for some $\lambda>0$ and fix any IC-IR social choice rule $\left(\hat{x}, \hat{y}\right): \Theta \rightarrow X \times Y$. Under the corresponding $\hat{\psi}_{distr}$, 
\begin{enumerate}[(i)]
\item there exists a sub-game perfect equilibrium that achieves intermediated implementation;
\item every sub-game perfect equilibrium achieves intermediated implementation if and only if $\hat{x}: \Theta \rightarrow X$ satisfies (\ref{eqn_du}) in a quasi-linear economy where agent utilities are given by $-\lambda^{-1}\log \left(\sum_{s=1}^S \theta_s \exp\left(-\lambda x_s\right)\right)-y$.
\end{enumerate}
\end{cor}

\begin{proof}
Rewrite (\ref{eqn_ic}) as $\tilde{v}\left(x\left(\theta\right),\theta\right)-y\left(\theta
\right)\geq \tilde{v}\left(x\left(\theta'\right),\theta\right)-y\left(\theta'\right)$ $\forall \theta'$, where $\tilde{v}\left(x,\theta\right) \coloneqq -\lambda^{-1}\log \left(\sum_{s=1}^S \theta_s \exp\left(-\lambda x_s\right)\right).$ The remainder of the proof is the same as that of Theorem \ref{thm_interdependent}.
\end{proof}

\subsection{Income Redistribution with Decentralized Contracting}\label{sec_tax}
There is a government, a finite number of firms and a continuum of workers. Each worker privately observes his innate ability $\theta$ and can work for at most one firm. An employment contract $\left(x,y\right)$ specifies an after-tax income $x \in X\subset \mathbb{R}$ and a required performance $y \in Y = \mathbb{R}$. It yields a utility $u\left(x,y,\theta\right)=x-v\left(y,\theta\right)$ to type $\theta$ workers and a profit $\pi(x,y,\theta)=h\left(y,\theta\right)-x$ to firms.
 
Whether firms have private or interdependent values depends on the contracting technology being used. The private-value case in which $y=h\left(y,\theta\right)$ represents the effective labor output is considered by  \cite{mirrlees}, whereas the interdependent-value case in which $y$ represents labor hours and $h\left(y,\theta\right)=\theta y$ is investigated by \cite{stantcheva}.

Consider first the benchmark case  in which the government owns firms and hence the contracting technology. The social choice rule $\left(\hat{x}, \hat{y}\right): \Theta \rightarrow X \times Y$ that maximizes her  redistributive objective, subject to workers' incentive constraints,\footnote{To be precise, \cite{mirrlees} ignores individual rationality constraints. Introducing them into the problem could affect the target social choice rule but not the effectiveness of our implementation strategy. } is called the \emph{constrained-efficient allocation} in the public finance literature. As for implementation, the government can propose the menu of target employment contracts $\left\{\left(\hat{x}\left(\theta\right), \hat{y}\left(\theta\right)\right): \theta \in \Theta \right\}$ to workers and let them self-select; after production takes place, she charges income tax $\hat{\pi}\left(\hat{x}\left(\theta\right)\right)$'s using the Mirrleesian tax schedule, leaving  firms break even contract by contract.

In today's market economies, it is firms, rather than the government, who design and enforce employment contracts. The government faces a partial control problem: on the one hand, it is difficult to regulate worker performance that depends on nonverifiable or even subjective information (\cite{levin}); on the other hand, it is much easier to regulate wage incomes, whose reporting is mandated by law.  

We examine the effectiveness of wage-income regulations in achieving constrained efficiency. Among the candidate policies defined in Section \ref{sec_candidatepolicy}, 

\begin{itemize}
\item $\psi_{per-unit}$ represents (nonlinear) income tax schedules that nest the Mirrleesian tax schedule  $\hat{\psi}_{per-unit}$ as a special case;

\item $\hat{\psi}_{distr}$ enforces the target income distribution among   employees.
\end{itemize}
According to our theorems, the effectiveness of income tax schedules depends on whether parties contract on effective labor outputs or labor hours. The Mirrleesian tax schedule remains effective in the first case, but is in general ineffective, along with other income tax schedules, in the second case.\footnote{\cite{stantcheva} characterizes the implementable outcomes by income tax schedules when the labor hours specified in the employment contract cannot be regulated by the government. The author adopts Miyazaki (1977)'s solution concept, which makes her results incomparable to ours.} Restoring the effectiveness of the Mirrleesian tax schedule in the second case requires the use of distribution regulations \`{a} la disability quota.\footnote{Quotas for hiring people with disabilities and the accompanying fines for non-compliance are commonly enforced in OECD countries (\cite{mont}) and are recently introduced to the U.S. by \href{https://www.disabilityscoop.com/2017/01/13/feds-set-disability-hiring-quota/23182/}{the Equal Employment Opportunity Commission}. 
A important subset of the broader disability category is self-reported disability, which includes attention deficit-disorder, back pain, diabetes, HIV/AIDS, etc.. \cite{diamondmirrlees} pioneer the idea that it is impossible to know whether a person is truly disabled, so ensuring the truth-telling of hidden disabilities is integral to the success of  redistribution systems. This idea serves as the basis for later developments in public finance, e.g., \cite{golosovtsyvinski}.} As for enforcement, frictions certainly exist in practice (e.g., the knowledge burden on the principal) but can be partly alleviated by the tricks developed in  Section \ref{sec_extension}.

\section{Conclusion}\label{sec_conclusion}
So far we have restricted intermediaries to producing homogeneous products and to competing \`{a} la Bertrand on the $y$-dimension. In reality, intermediaries can be horizontally differentiated or wield market power, and it is imperative to incorporate these considerations into future researches. Appendix \ref{sec_monopolistic} examines the case of monopolistic intermediary, showing that distribution regulation, plus rebates to agents, achieves intermediated implementation in every equilibrium.
 
Our result questions the effectiveness of certain pro-competition policies (e.g., banning resale price maintenance), showing that they may not alter the outcomes that the principal can implement through intermediaries after all. These outcomes are typically inefficient unless the principal is a benevolent social planner as in the healthcare and taxation applications.

To focus on the new implementation friction stemming from vertical restrictions and agent hidden characteristics, we abstract away from standard implementation frictions such as limited use of monetary transfers. We consider the weakest and strongest consumption regulations and give necessary and sufficient conditions for achieving intermediated implementation in every equilibrium. In the future, it will be interesting to consider in-between policies, or to characterize the implementable outcomes by the current policies when the condition for achieving full implementation fails to hold.

\appendix
\section{Omitted Proofs}
\subsection{Proofs of Section \ref{sec_private}}\label{sec_a1}
\paragraph{Proof of Lemma \ref{lem_payoffequivalence}}
\begin{proof}
Fix $\hat{\psi}_{per-unit}$ and an equilibrium it induces. Define \[A_i=\left\{\left(x,y\right): \left(x,y\right) \text{ is sold by intermediary } i \text{ to a positive measure of agents}\right\}\] and $A=\cup_{i=1}^I A_i$. Denote the measure of agents who purchase bundle $b \in A$ from intermediary $i$ by $\mu_i\left(b\right)$. Under the assumption of private values, no $b \in A$ incurs a loss because otherwise any intermediary $i$ with $\mu_i\left(b\right)>0$ can drop $b$ from its menu and save the loss. 

We proceed in two steps.

\paragraph{Step 1.} Show that every element of $A$ yields zero profit. 

\bigskip Suppose to the contrary that $A$ contains profitable elements. Since $X$ is a finite set, $A$ is, too, and therefore has a most profitable element denoted by $\left(x,y\right)$. By assumption, $\pi\left(x,y\right)-\hat{\pi}\left(x\right)>0$. 

Consider two cases. 

\paragraph{Case 1.} $\left(x,y\right) \notin A_i$ for some $i \in \left\{1,\cdots, I\right\}$. 

\bigskip

Consider a deviation by intermediary $i$ that adds $\left(x,y-\epsilon\right)$ to its menu. When $\epsilon$ is small but positive, the deviation attracts (1) all agents who used to purchase $\left(x,y\right)$ from other intermediaries than $i$, and (2) maybe some agents who used to purchase other bundles than $\left(x,y\right)$. When $\epsilon$ is sufficiently small, the first effect changes $i$'s profit by $\mu_{-i}\left((x,y)\right)\left[\pi\left(x,y-\epsilon\right)-\hat{\pi}\left(x\right)\right]>0$, whereas the second effect strictly increases $i$'s profit because $\left(x,y\right)$ is the most profitable bundle. Therefore, the deviation is profitable and the original outcome cannot arise in equilibrium.  

\paragraph{Case 2.}
$x \in A_i$ for every $i \in \left\{1,\cdots, I\right\}$. 

\bigskip

Consider a deviation by any intermediary $i$ that adds $\left(x,y-\epsilon\right)$ to its menu. In addition to the effects studied in Case 1, the deviation also attracts all agents who used to purchase $\left(x,y\right)$ from $i$. The resulting effect  $\mu_i\left((x,y)\right)\left[\pi\left(x,y\right)-\pi\left(x,y-\epsilon\right)\right]$ on $i$'s profit is negligible when $\epsilon$ is small.  

\paragraph{Step 2.} Show that all agents obtain target levels of utilities. 

\bigskip

From Step 1, we know that if the equilibrium utility of type $\theta$ agent falls short of the target level, then $\hat{x}\left(\theta\right)$ is not traded on the equilibrium path. But then any intermediary can add $\left(\hat{x}\left(\theta\right), \hat{y}\left(\theta\right)+\epsilon\right)$ to its menu and make a profit when $\epsilon$ is positive but small, which leads to a contradiction. 

On the other hand, if the equilibrium utility of type $\theta$ agent exceeds the target level, then let $\left(\hat{x}\left(\theta'\right),y\right)$ be the corresponding consumption bundle, and note that $u\left(\hat{x}\left(\theta'\right),y, \theta\right)>u\left(\hat{x}\left(\theta\right), \hat{y}\left(\theta\right),\theta\right) \geq u\left(\hat{x}\left(\theta'\right), \hat{y}\left(\theta'\right),\theta\right)$ (the first inequality is by assumption and the second one by incentive compatibility). Since $u$ is strictly decreasing in $y$ and $\pi$ is strictly increasing in $y$, it follows that $y<\hat{y}\left(\theta'\right)$ and $\pi\left(\hat{x}\left(\theta'\right),y\right)<\pi\left(\hat{x}\left(\theta'\right),\hat{y}\left(\theta'\right)\right)\coloneqq \hat{\pi}\left(\hat{x}\left(\theta'\right)\right)$, i.e., the bundle makes a loss, which is impossible. 
\end{proof}

\paragraph{Proof of Theorem \ref{thm_private}}
\begin{proof}
Part (i): Suppose to the contrary that we cannot sustain the strategy profile $\mathcal{M}_i^*=\left\{\left(\hat{x}\left(\theta\right), \hat{y}\left(\theta\right)\right): \theta \in \Theta \right\}$, $i=1,\cdots, I$ on the equilibrium path. Then there exists $\left(\hat{x}\left(\theta'\right), y\right)$ such that $u\left(\hat{x}\left(\theta'\right), y, \theta\right) \geq u\left(\hat{x}\left(\theta\right), \hat{y}\left(\theta\right), \theta\right)$ for some $\theta$, $\pi\left(\hat{x}\left(\theta'\right), y\right) \geq \hat{\pi}\left(\hat{x}\left(\theta'\right)\right)\coloneqq \pi\left(\hat{x}\left(\theta'\right), \hat{y}\left(\theta'\right)\right)$ and one of the inequalities is strict. Assume w.l.o.g. that it is $u\left(\hat{x}\left(\theta'\right), y, \theta\right) > u\left(\hat{x}\left(\theta\right), \hat{y}\left(\theta\right), \theta\right)$. Combining with incentive compatibility shows $u\left(\hat{x}\left(\theta'\right), y, \theta\right) > u\left(\hat{x}\left(\theta\right), \hat{y}\left(\theta\right), \theta\right) \geq u\left(\hat{x}\left(\theta'\right), \hat{y}\left(\theta'\right),\theta\right)$. Since $u$ is decreasing in $y$ and $\pi$ is increasing in $y$, it follows that $y<\hat{y}\left(\theta'\right)$ and 
$\pi\left(\hat{x}\left(\theta'\right), \hat{y}\left(\theta'\right)\right)>\pi\left(\hat{x}\left(\theta'\right), y\right) \geq \hat{\pi}\left(\hat{x}\left(\theta'\right)\right)\coloneqq \pi\left(\hat{x}\left(\theta'\right), \hat{y}\left(\theta'\right)\right) $, which is impossible. 

\bigskip 

\noindent Part (ii): Lemma \ref{lem_payoffequivalence} implies that if an equilibrium induced by $\hat{\psi}_{per-unit}$ fails to achieve intermediated implementation, then some target consumption good $\hat{x}\left(\theta\right)$ is not traded on the equilibrium path. But then type $\theta$ agents must be indifferent between $\left(\hat{x}\left(\theta\right), \hat{y}\left(\theta\right)\right)$ and another bundle prescribed by the target social choice rule, because otherwise any intermediary can add $\left(\hat{x}\left(\theta\right), \hat{y}\left(\theta\right)+\epsilon\right)$ to its menu and make a profit when $\epsilon$ is small but positive. 
\end{proof}

\subsection{Proofs of Section \ref{sec_interdependent}}\label{sec_a2}
\paragraph{Completing Example \ref{exm_interdependent}}
\begin{proof}
Take any other per-unit fee schedule $\psi_{per-unit}$ than $\hat{\psi}_{per-unit}$. In Section \ref{sec_interdependent_perunit}, we concluded that $\psi_{per-unit}$ could only achieve intermediated implementation  through cross subsidization. In what follows, we begin by assuming that this is the case and use $b^-$ (resp. $b^+ $) to denote the bundle that incurs a loss (resp. makes a  profit) under $\psi_{per-unit}$ (note that $b^-, b^+\in \left\{\left(\hat{x}_1, \hat{y}_1\right), \left(\hat{x}_2, \hat{y}_2\right)\right\}$). We then demonstrate that  $b^-$ and $b^+$ cannot both be traded in any equilibrium induced by $\psi_{per-unit}$  that achieves intermediated implementation, thus arriving at a contradiction. The conclusion is that  no per-unit fee schedule, whether it entails cross subsidization or not, achieves intermediated implementation in any equilibrium.

For each $b \in \left\{b^-, b^+ \right\}$, denote the measure of agents who purchase $b$ from intermediary $i$ by $\mu_i\left(b\right)$ and the net profit from serving $b$ to its target agents  by $\xi\left(b\right)$. By assumption, $\xi\left(b^+\right)>0$ and $\xi\left(b^-\right)<0$. 

Consider two cases.

\paragraph{Case 1.}$b^-=\left(\hat{x}_1, \hat{y}_1\right)$ and $b^+=\left(\hat{x}_2, \hat{y}_2\right)$. Consider two subcases.

\paragraph{Case 1(a).} Only one active intermediary (name it $i$) serves $b^-=\left(\hat{x}_1, \hat{y}_1\right)$ to type $\theta_1$ agents.

\bigskip

In this case, all other active intermediaries than $i$ must be serving $b^+=\left(\hat{x}_2, \hat{y}_2\right)$ to type $\theta_2$ agents and making profits. 
Consider a deviation by $i$ that adds $\left(\hat{x}_2, \hat{y}_2-\epsilon\right)$ to its menu. When $\epsilon$ is small but positive, the new bundle attracts all type $\theta_2$ agents in the economy while screening out all type $\theta_1$ agents, who strictly prefer $b^-=\left(\hat{x}_1, \hat{y}_1\right)$ to $b^+=\left(\hat{x}_2, \hat{y}_2\right)$ (the (IC$_{\theta_1}$) constraint is slack under the target social choice rule) and therefore find the new bundle unattractive. The resulting change in $i$'s profit $\mu_{-i}\left(b^+\right) \left[\xi\left(b^+\right)-\mathcal{O}\left(\epsilon\right)\right]-\mu_i\left(b^+\right) \mathcal{O}\left(\epsilon\right)$ is strictly positive.

\paragraph{Case 1(b).} Multiple active intermediaries serve  $b^-=\left(\hat{x}_1, \hat{y}_1\right)$ to type $\theta_1$ agents. In this case, any such intermediary can drop $b$ from its menu and save the loss.

\paragraph{Cases 2.} $b^-=\left(\hat{x}_2, \hat{y}_2\right)$ and $b^+=\left(\hat{x}_1, \hat{y}_1\right)$. Here we focus on the subcase in which only one active intermediary (name it $i$) serves $b^-=\left(\hat{x}_2, \hat{y}_2\right)$ to type $\theta_2$ agents. The proof of the other subcase is the same as that of Case 1(b) and is therefore omitted.

Suppose $i$ adds $\left(\hat{x}_1, \hat{y}_1-\epsilon\right)$ to its menu. When $\epsilon$ is small but positive, the new bundle attracts all type $\theta_1$ agents in the economy. As for type $\theta_2$ agents, they are indifferent between $b^-=\left(\hat{x}_2, \hat{y}_2\right)$ and $b^+=\left(\hat{x}_1, \hat{y}_1\right)$ (the (IC$_{\theta_2}$) constraint is binding under the target social choice rule) and therefore find the new bundle attractive, too. The resulting change in $i$'s profit is 
\begin{align*}
\tag{change caused by type $\theta_1$ agents}&\mu_{-i}\left(b^+\right)\left[\xi\left(b^+\right)-\mathcal{O}\left(\epsilon\right)\right]-\mu_{i}\left(b^+\right)\mathcal{O}\left(\epsilon\right)\\
+&\left[\mu_i\left(b^-\right)+\mu_{-i}\left(b^-\right)\right]\left[\xi\left(b^+\right)-\mathcal{O}\left(\epsilon\right)+\pi\left(b, \theta_2\right)-\pi\left(b,\theta_1\right)\right]\\
\tag{change caused by type $\theta_2$ agents}-&\mu_i\left(b^-\right)\xi\left(b^-\right),
\end{align*}
where the first and third lines are strictly positive because $\xi\left(b^+\right)>0>\xi\left(b^-\right)$, whereas the second line is strictly positive because $\pi\left(b,\theta_2\right)>\pi\left(b,\theta_1\right)$ for any $b \in \left\{b^-,b^+\right\}$. 
\end{proof}

\begin{lem}\label{lem_nicp}
Let $x: \Theta \rightarrow X$ be an implementable consumption rule that satisfies (\ref{eqn_du}). Then for any $\Theta' \subset \Theta$ such that $x\left(\theta\right) \neq x\left(\theta'\right)$ $\forall \theta, \theta' \in \Theta'$ and any cyclic permutation $\sigma$ of $\Theta'$,  $x \circ \sigma: \Theta' \rightarrow X$ is not implementable among the agents in $\Theta'$.
\end{lem}

\begin{proof}
Take any $\Theta'$ be as above and write $\Theta'=\left\{\theta_1,\cdots, \theta_m\right\}$. Also write $x\left(\theta_i\right)=x_i$ for $i=1,\cdots, m$. By assumption, there exists $y: \Theta' \rightarrow \mathbb{R}$ such that
\begin{align*}
v\left(x_2, \theta_1\right)-y_2 &\leq v\left(x_1, \theta_1\right)-y_1,\\
v\left(x_3, \theta_2\right)-y_3 &\leq v\left(x_2, \theta_2\right)-y_2,\\
&\vdots\\
v\left(x_1, \theta_m\right)-y_1 & \leq v\left(x_m, \theta_m\right)-y_m.
\end{align*}
Summing over these inequalities yields $\sum_{i=1}^{m} v\left(x_{i+1}, \theta_i\right) \leq \sum_{i=1}^m v\left(x_i,\theta_i\right)$. Meanwhile, if $\left(x_2,\cdots, x_m, x_1\right)$ is implementable among $\left(\theta_1,\cdots, \theta_m\right)$, then there exists $y': \Theta' \rightarrow \mathbb{R}$ such that \begin{align*}
v\left(x_2, \theta_1\right)-y_2' &\geq v\left(x_1, \theta_1\right)-y_1', \\
v\left(x_3, \theta_2\right)-y_3' &\geq v\left(x_2, \theta_2\right)-y_2', \\
&\vdots\\
v\left(x_1, \theta_m\right)-y_1' & \geq v\left(x_m, \theta_m\right)-y_m'.
\end{align*}
Summing over these inequalities yields $\sum_{i=1}^{m} v\left(x_{i+1}, \theta_i\right) \geq \sum_{i=1}^m v\left(x_i,\theta_i\right)$. Thus, a sufficient condition for $\left(x_2,\cdots, x_m, x_1\right)$ to be not implementable among $\left(\theta_1,\cdots, \theta_m\right)$ is (\ref{eqn_du}), which says that $\sum_{i=1}^{m} v\left(x_{i+1}, \theta_i\right) \neq \sum_{i=1}^m v\left(x_i, \theta_i\right)$. 
\end{proof}

\paragraph{Proof of Theorem \ref{thm_interdependent}(i)}
\begin{proof}
Take any strategy profile $\left(\mathcal{M}_i, \mathcal{M}_{-i}^*\right)$ whereby intermediary $i$ deviates unilaterally from offering $\left\{\left(\hat{x}\left(\theta\right), \hat{y}\left(\theta\right)\right): \theta \in \Theta\right\}$ to agents. Below we demonstrate that the deviation must be  unprofitable if agents break ties in favor of their target consumption goods off the equilibrium path. We begin by assuming that the deviation is profitable and show that this leads to a contradiction.

For each $j=1,\cdots, I$, let $x^j: \Theta \rightarrow X$ denote the consumption correspondence generated by intermediary $j$'s sold bundles and $x^j\left(\Theta\right)$ the image of $x^j$. For any $x, x' \in X$ such that $\hat{x}\left(\theta\right)=x$ and $x' \in x^j\left(\theta
\right)$ for some $\theta$ and $j$, we say that $x$ (resp. $x'$) is an \emph{immediate predecessor} (resp. \emph{immediate successor}) of $x'$ (resp. $x$), and write $x \rightarrow x'$. Let $y^j\left(x\right)$ denote the price charged by intermediary $j$ for any $x \in x^j\left(\Theta\right)$. Note that $y^j\left(x\right)=\hat{y}\left(x\right)$ for any $j \neq i$ and $x \in X$, where $\hat{y}\left(x\right)$ is the target price of $x$. 

Consider the graph generated by $\left(\mathcal{M}_i, \mathcal{M}_{-i}^*\right)$. For each $x \in x^i\left(\Theta\right)$, only three situations can happen: (1) $x$ has no immediate predecessor or successor, (2) $x$ is part of a cycle, and (3) $x$ is part of a chain. We now go through these cases one by one.

\paragraph{Case 1.} $x$ has no immediate predecessor or successor. In this case, we must have $y^i\left(x\right)=\hat{y}\left(x
\right)$, as the remaining possibilities can be ruled out in as follows. First, if $y^i\left(x\right)>\hat{y}\left(x\right)$, then no agent purchases $x$ from $i$, who must therefore violate the distributional requirement and incur a loss, a contradiction. Second, if $y^i\left(x\right)<\hat{y}
\left(x\right)$, then all agents with target consumption good $x$ purchase it from $i$. To meet the distributional requirement, $i$ must charge $y^i\left(x'\right) \leq \hat{y}\left(x'\right)$ for every $x' \in X-\left\{x\right\}$ and incur a loss, a contradiction. 

\paragraph{Case 2.} $x$ is part of a cycle $x \rightarrow x' \rightarrow \cdots$. Let $\theta, \theta', \cdots$ be any agent sequence consuming (1) $x$, $x'$, $\cdots$ under the target social choice rule, and (2) $x'$, $x''$, $\cdots$ under $\left(\mathcal{M}_i, \mathcal{M}_{-i}^*\right)$. If $\hat{x}: \Theta \rightarrow X$ satisfies (\ref{eqn_du}), then the new consumption rule that assigns $x'$ to $\theta$, $x''$ to $\theta'$, so on and so forth, cannot be part of any incentive compatible allocation and therefore cannot arise under $\left(\mathcal{M}_i, \mathcal{M}_{-i}^*\right)$. If $\hat{x}: \Theta \rightarrow X$ violates (\ref{eqn_du}), then the inequalities in the proof of Lemma \ref{lem_nicp} are all binding, implying that $\theta$ is indifferent between $x$ and $x'$, $\theta'$ is indifferent between $x'$ and $x''$, so on and so forth. This, together with the assumption that agents break ties in favor of their target consumption goods, implies that the cycle $x \rightarrow x' \rightarrow \cdots$ cannot arise under $\left(\mathcal{M}_i, \mathcal{M}_{-i}^*\right)$.

\paragraph{Case 3.} $x$ is part of a chain. Take any chain with end node $x'= \hat{x}\left(\theta'\right)$, and let $x''= \hat{x}\left(\theta''\right)$ be an immediate predecessor of $x'$. Note the following:
\begin{enumerate}[(1)]
\item incentive compatibility implies $v\left(x'', \theta'' \right) - \hat{y}\left(x''\right) \geq v\left(x', \theta''\right)-\hat{y}\left(x'\right)$;
\item $x'' \rightarrow x'$ implies $v\left(x'', \theta''\right)-\hat{y}\left(x''\right) \leq v\left(x', \theta''\right)-\min \left\{ y^i\left(x'\right), \hat{y}\left(x'\right)\right\}$;
\item the definition of end node implies that all type $\theta'$ agents consume $x'$. 
\end{enumerate}
Combining (1)-(3) yields $y^i\left(x'\right)=\hat{y}\left(x'\right)$ (hence $y^i\left(x\right)=\hat{y}\left(x\right)$ for every $x$ in the chain by induction), as the remaining possibilities can be ruled out as follows. First, if $y^i\left(x'\right)>\hat{y}\left(x'\right)$, then (2) becomes $v\left(x'', \theta''\right)-\hat{y}\left(x''\right) \leq  v\left(x', \theta''\right)-\hat{y}\left(x'\right) \leq v\left(x'', \theta''\right)-\hat{y}\left(x''\right)$, where the last inequality follows from incentive compatibility. Thus type $\theta''$ agents are indifferent between $x''$ and $x'$, and $x'' \rightarrow x'$ is incompatible with their off-path tie-breaking rule. Second, if $y^i\left(x'\right)<\hat{y}\left(x'\right)$, then the fact that $x'$ is an end node and  $x'' \rightarrow x'$ implies that all type $\theta'$ agents and some type $\theta''$ agents purchase $x'$ from intermediary $i$. Thus $i$ must violate the distributional requirement and incur a loss, a contradiction. 

\bigskip

The above argument shows that $y^i\left(x\right)=\hat{y}\left(x\right)$ for any $x \in x^i\left(\Theta\right)$. Then from $x^i\left(\Theta\right)=X$ (to meet the distributional requirement and make a profit), it follows that $\mathcal{M}_i=\left\{\left(\hat{x}\left(\theta\right), \hat{y}\left(\theta\right)\right): \theta \in \Theta \right\}$, a contradiction.
\end{proof}

\paragraph{Proof of Theorem \ref{thm_interdependent}(ii)}
\begin{proof}
\emph{``If'' direction} \quad Take any equilibrium strategy profile and let $x^i$ denote the consumption correspondence generated by intermediary $i$'s sold bundles. Note that $\cup_{i=1}^I x^i\left(\Theta\right)=X$, because otherwise either some intermediary violates the distributional requirement or all intermediaries are inactive, and neither situation can arise in equilibrium. 

Take the graph generated by the equilibrium strategy profile (as defined in the proof of Theorem \ref{thm_interdependent}(i)). Clearly, the graph contains no cycle because of (\ref{eqn_du}), and it contains no chain because otherwise some intermediary must violate the distributional requirement and would rather exit the market.  Thus the target consumption rule gets implemented, and the competition between intermediaries drives prices to target levels.

\bigskip

\noindent \emph{``Only if'' direction} \quad  Let $\Theta'$ denote a generic  subset of $\Theta$ such that $\hat{x}(\theta)\neq \hat{x}(\theta')$ for any $\theta, \theta' \in \Theta'$ and $\sum_{\theta \in \Theta'} v(\hat{x}(\theta),\theta)=\sum_{\theta \in \Theta'} v(\hat{x}\left(\sigma\left(\theta\right)\right),\theta)$ for some cyclic permutation $\sigma:\Theta' \rightarrow \Theta'$. Hereafter $\Theta'$ is called an  \emph{indifference cycle}. By definition, (1) $\hat{x}$ violates (\ref{eqn_du}) if and only if indifference cycle exists, and (2) any implementable consumption correspondence that differs from the target consumption rule and satisfies the distributional requirement can be obtained from permuting  consumption goods along indifference cycles. 

Take any consumption correspondence as in point (2). Combining it with any transfer rule $y=\hat{y}+d$ where $d \in \mathbb{R}$ is  independent of $\theta$ yields an incentive compatible social choice correspondence. Let $\hat{d}$ be the constant such that the resulting social choice correspondence yields zero profit under $\hat{\psi}_{distr}$. To sustain this social choice correspondence in an equilibrium, suppose that agents adopt the same tie-breaking rule on and off the equilibrium path. Then from point (2) of the previous paragraph, it follows that off the equilibrium path, agents must consume the same consumption goods as they do on the equilibrium path in order for the deviating intermediary to meet the distributional requirement. But then the deviation must be unprofitable, because no one can undercut $\hat{y}+\hat{d}$ and still make a profit. 
\end{proof}


\section{Monopolistic Intermediary}\label{sec_monopolistic}
This appendix examines the case of monopolistic intermediary $I=1$. To obtain the sharpest insights, suppose $X, \Theta$ are closed intervals, and the distribution of agent type has a positive density on $\Theta$. The function $v\left(x,\theta\right)$ is  continuously differentiable and satisfies the single-crossing property.

For starters, note that per-unit fee schedules cannot generally deter the intermediary from performing monopolistic screening among agents, i.e., distort the consumption of some agent in order to extract information rents from other agents.\footnote{The resulting vertical externality differs from double marginalization, because double marginalization requires downward sloping demand curves to work, whereas our agents have unit demand for consumption bundles.} To make progress, let $\left\{\left(\hat{x}\left(\theta\right), \overline{y}\left(\theta\right)\right): \theta \in \Theta\right\}$ be the IC-IR allocation that implements the target consumption rule and charges the maximal price for every  consumption good. By the envelope theorem (\cite{milgromsegal}), 
\[\overline{y}\left(\theta\right)=v\left(\hat{x}\left(\theta\right),\theta\right)-\int_{\underline{\theta}}^{\theta} v_{\theta}\left(\hat{x}\left(s\right), s\right)ds.\]
For each $x \in X$, let $\overline{y}\left(x\right)$ be the unique price such that $\left(x,\overline{y}\left(x\right)\right) \in \left(\hat{x}, \overline{y}\right)\left(\Theta\right)$ and $\overline{\pi}\left(x\right)$ be the expected profit from serving $\left(x,\overline{y}\left(x\right)\right)$ to $\left\{\theta\in \Theta: \left(\hat{x}, \overline{y}\right)\left(\theta\right)=\left(x, \overline{y}\left(x\right)\right)\right\}$. Refine the distribution regulation as follows: 
\[\overline{\psi}_{distr}\left(\nu\right)=\begin{cases}
\int_{x \in X} \overline{\pi}\left(x\right) d\nu &\text{ if } \nu=\hat{P}_x,\\
+\infty &\text{ otherwise},
\end{cases}\] 
where $\nu$ denotes the measure on $X$ induced by the sold bundles.

\begin{thm}\label{thm_monopoly} 
For any IC-IR social choice rule $\left(\hat{x},\hat{y}\right):\Theta \rightarrow X \times Y$, enforcing the corresponding distribution regulation $\overline{\psi}_{distr}$ on the monopolistic intermediary and paying a uniform rebate $\overline{y}\left(\underline{\theta}\right)-\hat{y}\left(\underline{\theta}\right)$ to all agents achieve  intermediated implementation in every equilibrium.
\end{thm}

\begin{proof}
Take any social choice correspondence $\left(x^*, y^*\right)$ that can arise in equilibrium. By \cite{milgromshannon}, any implementable consumption correspondence is nondecreasing in $\theta$, which combined with $\nu=\hat{P}_x$ yields $x^*=\hat{x}$ and hence $y^*=\overline{y}$. Meanwhile, applying the envelope theorem to the target social choice rule yields
$
\hat{y}\left(\theta\right)=v\left(\hat{x}\left(\theta\right),\theta\right)-\int_{\underline{\theta}}^{\theta} v_{\theta}\left(\hat{x}\left(s\right),s\right)ds +  \text{constant}
$
and hence $\overline{y}\left(\theta\right)-\hat{y}\left(\theta\right)=\overline{y}\left(\underline{\theta}\right)-\hat{y}\left(\underline{\theta}\right)$ for any $\theta$,\footnote{When agent types are discrete as in the baseline model, we can still use distribution regulation to implement the target consumption rule (assuming (\ref{eqn_du})), but can no longer use the envelope theorem to pin down $\overline{y}\left(\theta\right)-\hat{y}\left(\theta\right)=\overline{y}\left(\underline{\theta}\right)-\hat{y}\left(\underline{\theta}\right)$ for all $\theta$. Rather than a uniform rebate, we should give personalized rebate equal to $\overline{y}\left(\theta\right)-\hat{y}\left(\theta\right)$ if $x=\hat{x}\left(\theta\right)$.} so the uniform rebate $\overline{y}\left(\underline{\theta}\right)-\hat{y}\left(\underline{\theta}\right)$ restores the net prices paid by agents to the target levels.
\end{proof}

Theorem \ref{thm_monopoly} shows that distribution regulation, plus a uniform rebate to agents, achieve intermediated implementation through a monopolistic intermediary. Besides the aforementioned applications, this result explains why authorities impose strict  rules on the housing types that major real-estate companies can develop for low-income households. It suggests that one way to stop schools from oversubscribing students into the \href{https://www.air.org/edsector-archives/blog/fraud-lunchroom}{National School Lunch Program} is to cap the number of eligible students based on the income distribution within the school district.

\end{document}